\documentclass[copyright,creativecommons]{packages/eptcs}
\providecommand{\event}{GandALF 2011} 

\usepackage{breakurl}
\usepackage{alltt}
\usepackage{url}
\usepackage{xspace}
\usepackage{graphicx}
\usepackage{subfig}
\usepackage{amsmath}
\usepackage{latexsym}
\usepackage{amssymb}
\usepackage{pstricks,pst-node,pst-text}
\usepackage{tikz}
\usepackage{packages/algorithm, packages/algorithmic, packages/algorithmic-fix}
\usepackage{csquotes}
\usepackage{multirow}

\usepackage{color}

\newcommand{ \lin}{\mbox{L}}
\newcommand{\ml}{\mbox{ML}}
\newcommand{\m}{{\mbox{d}}}
\newcommand{\bd}{{\mbox{Bd}}}

\newcommand{\p}{{\mbox{P}}}
\newcommand{\br}{{\mbox{Bd}}}
\newcommand{\traj}{\tau}

\newcommand{\x}{{\mbox{x}}}

\newcommand{\cycle}{\zeta}
\newcommand{\Cycle}{\mbox{C}}

\newcommand{\reach}{{\mbox{Reach}}}
\newcommand{\dist}{{\mbox{Dist}}}
\newcommand{\sign}{{\mbox{Sign}}}
\newcommand{\D}{{\mbox{d}}}  
 
\newcommand{\seq}{\overline{\star}}  
\newcommand{\sgn}{\Sigma} 
\newcommand{\ssgn}{\overline{\sgn}}  
  
\newcommand{\Trj}{\mbox{Tr}}
\newcommand{\STrj}{\overline{\Trj}} 
\newcommand{\Hc}{\Cycle^h} 
 
\newcommand{\ssignat}{\sgn^s}
\newcommand{\be}{\mbox{e}}

\newtheorem{theorem}{Theorem}[section]
\newtheorem{lemma}[theorem]{Lemma}
\newtheorem{definition}[theorem]{Definition}

\newenvironment{proof}[1][Proof]{\begin{trivlist}
\item[\hskip \labelsep {\bfseries #1}]}{\end{trivlist}}

\newenvironment{newalgo}[2]{\begin{algorithm}
\caption{\textsc{#1}}\label{#2}
\begin{algorithmic}[1]
\vspace{0.0in}} {\end{algorithmic}\end{algorithm}}

\begin{document}

\title{Deciding Reachability 
 for 3-Dimensional\\ Multi-Linear Systems \thanks{This work
      was supported in part by the ``Concept for the Future'' of  
Karlsruhe Institute of
      Technology within the framework of the German Excellence  
Initiative.}}

\author{Olga Tveretina
\institute{Karlsruhe Institute of Technology\\ Karlsruhe, Germany}
\email{olga.tveretina@kit.edu}
\and
\qquad Daniel Funke
\institute{Karlsruhe Institute of Technology\\ Karlsruhe, Germany}
\email{\quad daniel.funke@student.kit.edu}
}
\def\titlerunning{Deciding Reachability  for 3-Dimensional Multi-Linear Systems}
\def\authorrunning{O. Tveretina, D. Funke}

\maketitle

\begin{abstract}
This paper deals with the problem of point-to-point reachability in multi-linear systems. 
These systems consist of a partition of the  Euclidean space into a finite number of regions  and a constant derivative assigned to each region in the partition, 
which governs the dynamical behavior of the system within it. The reachability problem for multi-linear systems has been proven to be decidable for the two-dimensional case  and  undecidable for  the dimension three and higher.
Multi-linear systems however exhibit certain properties that make  them very suitable for topological analysis. 
We prove that reachability can be decided exactly in the 3-dimensional case when  systems satisfy certain conditions. 
We show with experiments that our approach can be orders of magnitude more efficient  than simulation.
\end{abstract}

\section{Introduction}

During the last decades a lot of devices have been developed  that consist of computers  interacting with a physical environment. Computers perform discrete operations, while a physical environment has continuous dynamics. Such systems are called hybrid systems.
Many of the  applications of hybrid systems, such as
intelligent highway systems, air traffic management systems and others  are \emph{safety critical} and require
the guarantee of a safe operation.

Formally verifying safety properties of hybrid
systems consists of building a set of reachable states and checking if this
set intersects with a set of unsafe states. Therefore one of the
most fundamental problems in the analysis of hybrid systems is the
\emph{reachability} problem.

The reachability problem is known as being difficult. It has been shown to be
decidable for special kinds of hybrid automata \cite{AD94,HKPV95,LPS00,LPY99,LPY99-2} including timed automata \cite{AD94}, 
some classes of rectangular hybrid automata \cite{HKPV95} and o-minimal hybrid automata \cite{LPS00}.

Since only  certain kinds of hybrid systems allow  for the exact computation of the reachable set,    approaches for safety verification include  the approximation of reachability analysis and  abstraction techniques. 
But these  techniques are easy to fail when applied to large systems since the complexity rises up very quickly with an increase in system size.

One of the drawbacks of approximation and propagation techniques is that too little attention is paid to the geometric properties of the systems under analysis \cite{ASY07}.  There are two main approaches in this direction: 1) methods that use topological properties of the plane \cite{MP93}, and 2) techniques based on the existence of integrals and the ability to compute them \cite{B99}. 

In this paper we consider {\it multi-linear systems} (\ml) also often called  {\it piecewise constant derivative systems} (PCDs) in the literature.  
They  are a special kind of  hybrid system, where the number of dimensions  refers to the number of continuous 
variables. Such systems satisfy the following restrictions: A discrete state
is defined by a set of linear constraints and 
 discrete transitions do not change continuous variables. 
\ml~systems have been proven to be decidable for the two-dimensional case \cite{MP93}, 
whereas the results presented in \cite{AMP95}
state that such systems are  undecidable for  the dimension three and higher.

The decidability results for the $2$-dimensional case rely on the existence of a periodic trajectory after a finite number of steps. This property does not hold for higher dimensions.
Nevertheless, $3$-dimensional systems also feature some sort of regularity. And, as in the $2$-dimensional case, $3$-dimensional multi-linear systems  exhibit  certain properties that make  them very suitable for topological analysis.

{\it Contribution}. 
We  consider  a subclass of multi-linear systems,  that we call multi-linear $\lambda$-systems. These systems 
 satisfy the following property: If there is a cyclic trajectory, then  the the points of each cycle iteration intersecting the same boundary element of a polyhedron lie on a straight line. A straightforward consequence of this assumption is that the distances between the corresponding boundary points of different rounds are proportional ($\lambda$-property). 
We  introduce the notion of a hypercycle, a generalization of a cycle. The infinity criterion for $2$-dimensional case, refer to \cite{MP93}, has an analog in $3$ dimensions. We show  that the $\lambda$-property holds also for hypercycles, and the reachability can be decided exactly if the derived infinity criterion for $3$ dimensions holds for a hypercycle. 

We have implemented our approach and compared it with simulation.
As soon as our algorithm detects a cycle (or a hypercycle) for which the given infinity criterion holds, the algorithm requires constant number of steps. While the number steps for simulation grows exponentially with the distance between points. 
Algorithms for  computing  reachable states  are often based on floating point computations  that involve rounding errors and the correctness of such algorithms can be violated. Since our algorithm takes significantly less steps, it leads to more exact computations.

A complete version of the paper containing all proofs and the details of the benchmarks is available at 
\cite{TF11-2}.

\section{Multi-Linear  Systems}\label{sec:basics}

Multi-linear systems consist of a partition of the  Euclidean space into a finite number of regions  and a constant derivative
assigned to each region in the partition. 
In this section we define these systems in a way similar to \cite{MP93}.

We consider an $n$-dimensional Euclidean space  ${\cal{R}}^{n}$ with a metric $\m$ and points in it denoted by $x$ and $y$. In the following, we  specify the position of any point in $3$-dimensional space by three Cartesian coordinates. 
A {\em linear half space} is a set of all point in ${\cal{R}}^{n}$ satisfying $Ax+B \bowtie 0$, where $\bowtie\,\in \{<,\leq,>,\geq\}$, $A$ is a rational vector and $B$ is a rational number.
A polyhedron is a  subset  of ${\cal{R}}^{n}$ obtained by  intersecting a finite number of linear half spaces. 
Since we have  a finite number of linear half spaces that divide the complete $n$-dimensional Euclidean space, there are polyhedra that are not bounded from all sides. 



\begin{definition}[Polyhedral partition] Given a finite set of linear half spaces 
${\cal S}=\{A_ix+B_i \bowtie 0, 1\leq i \leq n\}$, we say that ${\cal{P}}({\cal S})=\{\p_1, \dots, \p_m\}$ is a polyhedral partition of  ${\cal{R}}^{n}$  by ${\cal S}$ if:  
$1)$ $\bigcup_{i=1}^m \p_i={\cal{R}}^{n}$, and $2)$   $\p_i\cap\p_j=\emptyset$ for distinct $i,j\in\{1,\dots, m\}$.
\end{definition}
When it is convenient we will use $\cal{P}$ instead of ${\cal{P}}({\cal S})$ to denote a polyhedral partition. 
Given a polyhedral partition $\cal{P}({\cal S})$,  we define  the  set of its  boundary points  as  
\[\br({\cal{P}}({\cal S}))=\{y \in{\cal{R}}^{n} ~|~\exists (Ax+B \bowtie 0)\in {\cal S}: Ay+B=0\}.\]

For each polyhedron $P\in {\cal P}({\cal S})$, we define the set of the boundary points as \[\br(P)=\br({\cal{P}}({\cal S}))\cap P.\]
Note that, depending on the partition,   the set of boundary points of  some polyhedra can be empty. 

\begin{definition}[Boundary element]
 Given a polyhedral partition ${\cal P} ({\cal S})$ and a polyhedron $P\in {\cal P} ({\cal S})$, we say that $\be$ is a boundary element of $P$ if the following holds.
\begin{enumerate}
 \item[(1)] $\be\subseteq \bd(P)$, and  
 \item[(2)] There is  $(Ax+B \bowtie 0)\in {\cal S}$ such that if  $y\in \be$ then $Ay+B=0$.
\end{enumerate}

\end{definition}

An $n$-dimensional  multi-linear system consists of a partitioning ${\cal{P}}({\cal S})=\{\p_1, \dots, \p_m\}$ of the space ${\cal{R}}^{n}$
 into a finite set of polyhedral regions  and a constant derivative $c_i$  assigned to each region $P_i$. 
We define such systems and a trajectory similar to \cite{MP93}.

\begin{definition}[Multi-linear system] We define a  multi-linear system   on ${\cal{R}}^{n}$ as
a pair ${\cal H}=({\cal P},f)$, where $\cal P$ 
is a polyhedral partition of ${\cal{R}}^{n}$ and  $f:{\cal P}\rightarrow {\cal{R}}^{n}$ is a function that assigns a vector $c$ to each $P\in\cal P$.  
\end{definition}

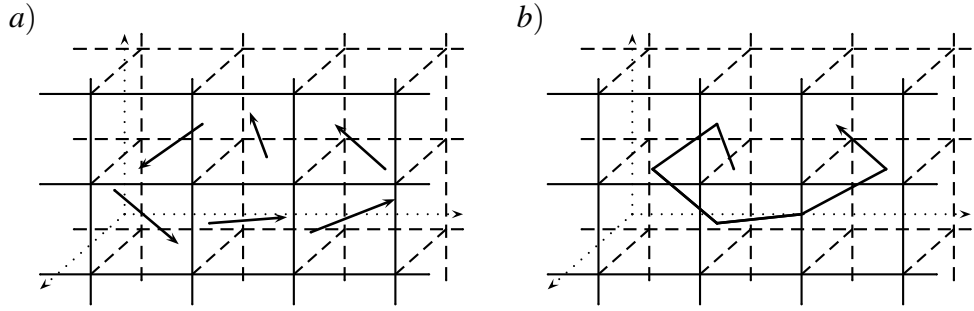
\begin{figure}
\begin{center}

 \psset{xunit=0.45,yunit=0.40}
\begin{pspicture}(0,0)(22,10)

\rput(-2,9.5){\large $a)$}

\rput(13,9.5){\large $b)$}

\psline[linestyle=dotted,arrows=c->](1,3)(11,3)
\psline[linestyle=dotted,arrows=c->](1,3)(1,9)
\psline[linestyle=dotted,arrows=c->](1,3)(-1.5,0.5)

\psline[arrows=c-c](0,0)(0,7.5)
\psline[arrows=c-c](3,0)(3,7.5)
\psline[arrows=c-c](6,0)(6,7.5)
\psline[arrows=c-c](9,0)(9,7.5)

\psline[arrows=c-c](-1.5,1)(10,1)
\psline[arrows=c-c](-1.5,4)(10,4)
\psline[arrows=c-c](-1.5,7)(10,7)

\psline[linestyle=dashed,arrows=c-c](0,1)(1.5,2.5)
\psline[linestyle=dashed,arrows=c-c](3,1)(4.5,2.5)
\psline[linestyle=dashed,arrows=c-c](6,1)(7.5,2.5)
\psline[linestyle=dashed,arrows=c-c](9,1)(10.5,2.5)

\psline[linestyle=dashed,arrows=c-c](0,4)(1.5,5.5)
\psline[linestyle=dashed,arrows=c-c](3,4)(4.5,5.5)
\psline[linestyle=dashed,arrows=c-c](6,4)(7.5,5.5)
\psline[linestyle=dashed,arrows=c-c](9,4)(10.5,5.5)

\psline[linestyle=dashed,arrows=c-c](0,7)(1.5,8.5)
\psline[linestyle=dashed,arrows=c-c](3,7)(4.5,8.5)
\psline[linestyle=dashed,arrows=c-c](6,7)(7.5,8.5)
\psline[linestyle=dashed,arrows=c-c](9,7)(10.5,8.5)

\psline[linestyle=dashed,arrows=c-c](-0.5,2.5)(11,2.5)
\psline[linestyle=dashed,arrows=c-c](-0.5,5.5)(11,5.5)
\psline[linestyle=dashed,arrows=c-c](-0.5,8.5)(11,8.5)

\psline[linestyle=dashed,arrows=c-c](1.5,0.8)(1.5,9)
\psline[linestyle=dashed,arrows=c-c](4.5,0.8)(4.5,9)
\psline[linestyle=dashed,arrows=c-c](7.5,0.8)(7.5,9)
\psline[linestyle=dashed,arrows=c-c](10.5,0.8)(10.5,9)

\psline[linewidth=0.035,arrows=c->](5.2,4.9)(4.7,6.4)
\psline[linewidth=0.035,arrows=c->](3.3,6)(1.4,4.5)
\psline[linewidth=0.035,arrows=c->](0.7,3.8)(2.6,2)

\psline[linewidth=0.035,arrows=c->](3.5,2.7)(5.8,2.9)
\psline[linewidth=0.035,arrows=c->](6.5,2.4)(9,3.5)
\psline[linewidth=0.035,arrows=c->](8.7,4.5)(7.2,6)


\psline[linewidth=0.035,arrows=c-c](19,4.5)(18.5,6)

\psline[linewidth=0.035,arrows=c-c](18.5,6)(16.6,4.5)

\psline[linewidth=0.035,arrows=c-c](16.6,4.5)(18.5,2.7)
\psline[linewidth=0.035,arrows=c-c](16.6,4.5)(18.5,2.7)

\psline[linewidth=0.035,arrows=c-c](18.5,2.7)(21,3)

\psline[linestyle=dotted,arrows=c->](16,3)(26,3)
\psline[linestyle=dotted,arrows=c->](16,3)(16,9)

\psline[linestyle=dotted,arrows=c->](16,3)(13.5,0.5)

\psline[arrows=c-c](15,0)(15,7.5)
\psline[arrows=c-c](18,0)(18,7.5)
\psline[arrows=c-c](21,0)(21,7.5)
\psline[arrows=c-c](24,0)(24,7.5)

\psline[arrows=c-c](13.5,1)(25,1)
\psline[arrows=c-c](13.5,4)(25,4)
\psline[arrows=c-c](13.5,7)(25,7)

\psline[linestyle=dashed,arrows=c-c](15,1)(16.5,2.5)
\psline[linestyle=dashed,arrows=c-c](18,1)(19.5,2.5)
\psline[linestyle=dashed,arrows=c-c](21,1)(22.5,2.5)
\psline[linestyle=dashed,arrows=c-c](24,1)(25.5,2.5)

\psline[linestyle=dashed,arrows=c-c](15,4)(16.5,5.5)
\psline[linestyle=dashed,arrows=c-c](18,4)(19.5,5.5)
\psline[linestyle=dashed,arrows=c-c](21,4)(22.5,5.5)
\psline[linestyle=dashed,arrows=c-c](24,4)(25.5,5.5)

\psline[linestyle=dashed,arrows=c-c](15,7)(16.5,8.5)
\psline[linestyle=dashed,arrows=c-c](18,7)(19.5,8.5)
\psline[linestyle=dashed,arrows=c-c](21,7)(22.5,8.5)
\psline[linestyle=dashed,arrows=c-c](24,7)(25.5,8.5)

\psline[linestyle=dashed,arrows=c-c](14.5,2.5)(26,2.5)
\psline[linestyle=dashed,arrows=c-c](14.5,5.5)(26,5.5)
\psline[linestyle=dashed,arrows=c-c](14.5,8.5)(26,8.5)

\psline[linestyle=dashed,arrows=c-c](16.5,0.8)(16.5,9)
\psline[linestyle=dashed,arrows=c-c](19.5,0.8)(19.5,9)
\psline[linestyle=dashed,arrows=c-c](22.5,0.8)(22.5,9)
\psline[linestyle=dashed,arrows=c-c](25.5,0.8)(25.5,9)

\psline[linewidth=0.035,arrows=c-c](18.5,2.7)(21,3)
\psline[linewidth=0.035,arrows=c-c](21,3)(23.5,4.5)
\psline[linewidth=0.035,arrows=c->](23.5,4.5)(22,6)

\end{pspicture}\caption{\label{Ex_trajectory}  
a) A simple 3-\ml~and b) a possible trajectory}
\end{center}
\end{figure}

In the following we concentrate on $3$-dimensional multi-linear  systems. A simple $3$-dimensional multi-linear system  is depicted in Figure \ref{Ex_trajectory}. 
 The trajectories of such systems are sequences of line segments, where the break points belong to the boundaries of  polyhedra. 
Multi-linear systems are deterministic in a sense that for  each initial point there is exactly one corresponding trajectory. 

We assume that the assigned derivative vectors of two neighboring polyhedra may not be directed towards the same boundary,
since this would lead to Zeno behavior when   a system performs infinitely many transitions in a finite period of time.

In the rest of the paper we use the following notations. By $\epsilon$ we denote the empty sequence. 
We use $s_1.s_2$ to denote the concatenation of sequences $s_1$ and $s_2$,   $(s^i)_{i=1}^m$ is a shortcut for the sequence $s^1.\dots.s^m$. Given a sequence $s$, we denote by $s^{\seq}$  a (possibly infinite) sequence $s. s. \dots$ if  $s$ is repeated at least two times. By $s_1\sqsubset s_2$ we mean that $s_2=s_2'.s_1.s_2''$ for some sequences $s_1,s_2,s_2',s_2''$ and at most one of $s_2'$ and $s_2''$ is not the empty sequence. 

In the following definitions  for simplicity and without loss of generality we can assume that a trajectory always starts at a boundary element. 

\begin{definition}[Trajectory ] Let  $\cal{H}$ be a \ml, and $x_0\in \cal{X}$ be a point.
\begin{enumerate} 
\item A trajectory starting at  $x_0$ is a  sequence  $\traj=x_0.x_1.x_2. \dots$
where for  $i\geq 0$ there is  $P_i\in \cal P$ such that   $x_i\in \bd(P_i)$ and for  $y\in ]x_{i-1},x_{i}[$ there is no $P'\in \cal P$ such that   $y\in \bd(P')$. We denote by $\Trj(\cal H)$ the set of all trajectories of $\cal{H}$. 

\item  A sub-trajectory of $\traj$, written as $\traj^s\sqsubset\traj$,  is a finite (possibly empty) sequence 
$\traj^s= x_i. x_{i+1}.\dots. x_{j}$.  We denote by $\Trj^s$ the set of all sub-trajectories of $\cal{H}$. 

\end{enumerate}
\end{definition}

\begin{definition}[Signature of a trajectory]
Let  $\cal{H}$ be a \ml, and $x_0\in \cal{X}$ be a point. We assume a trajectory  $\traj=x_0.x_1.x_2. \dots$ .  We say that a sequence of boundary elements $\sigma(\traj)=\be_0. \be_1. \be_2. \dots$ is  a signature of $\traj$ if  $x_i\in\be_i$ for $i\geq 0$.  We denote by $\sgn(\cal{H})$ the set of signatures of all trajectories of $\cal{H}$ and by $\ssignat(\cal H)$ the set of signatures of all sub-trajectories of $\cal{H}$.
 \end{definition}

\begin{figure}
\begin{center}
\begin{tikzpicture}[scale=0.40]

\draw[thick, dotted] (0,0) -- (6,0);

\draw[->,thick, color=black!100] (0,0) -- (0.5,0.7) -- (0.7,-0.4) -- (1,0);
\draw[->,thick, color=black!80] (1,0) -- (1.5,0.7) -- (1.7,-0.4) -- (2,0);
\draw[->,thick, color=black!60] (2,0) -- (2.5,0.7) -- (2.7,-0.4) -- (3,0);


\draw[->,thick, color=black!100] (5,0) -- (5.5,0.7) -- (5.7,-0.4) -- (7,1);

\draw[thick, dotted] (7,1) -- (3,5) ;

\draw[->,thick,color=black!100] (7,1) -- (8,0.8) -- (7.7,1.5) -- (6.5,1.5);
\draw[->,thick,color=black!80] (6.5,1.5) -- (7.5,1.3) -- (7.2,2) -- (6,2);
\draw[->,thick,color=black!60] (6,2) -- (7,1.8) -- (6.7,2.5) -- (5.5,2.5);


\draw[->,thick,color=black!100] (3.5,4.5) -- (4.5,4.3) -- (4.2,5) -- (2,5);

\draw[thick, dotted] (2,5) -- (-4,-1) ;

\draw[->,thick,color=black!100] (2,5) -- (2.1,4.2) -- (1.4,4.8) -- (0.8,4.8)  -- (1,4);
\draw[->,thick,color=black!80] (1,4) -- (1.1,3.2) -- (0.4,3.8) -- (-0.2, 3.8)  -- (0,3);
\draw[->,thick,color=black!60] (0,3) -- (0.1,2.2) -- (-0.6,2.8) -- (-1.2, 2.8)  -- (-1,2);

\draw[->,thick,color=black!100] (-3,0) -- (-3.1,-0.7) -- (-3.6,-0.2) -- (-4.2, -0.2)  -- (-4,-1);


\draw[thick, dotted] (-4,-1) -- (9,-1);

\draw[->,thick, color=black!100] (-4,-1) -- (-3.5,-0.3) -- (-3.3,-1.4) -- (-3,-1);
\draw[->,thick, color=black!80] (-3,-1) -- (-2.5,-0.3) -- (-2.3,-1.4) -- (-2,-1);
\draw[->,thick, color=black!60] (-2,-1) -- (-1.5,-0.3) -- (-1.3,-1.4) -- (-1,-1);

\draw[->,thick, color=black!100] (8,-1) -- (8.5,-0.3) -- (8.7,-1.4) -- (10,0);

\draw[thick, dotted] (10,0) -- (3,7);

\draw[->,thick,color=black!100] (10,0) -- (11,-0.2) -- (10.7,0.5) -- (9.5,0.5);
\draw[->,thick,color=black!100] (3.5,6.5) -- (4.4,6.3) -- (4.2,7) -- (2,7);


\draw[thick, dotted] (2,7) -- (-5,0);
\draw[->,thick,color=black!100] (2,7) -- (1.9,6.3) -- (1.4,6.8) -- (0.8,6.8)  -- (1,6);

\end{tikzpicture}\label{HC}
\caption{A schematic representation  of a  hypercycle}
\end{center}
\end{figure}
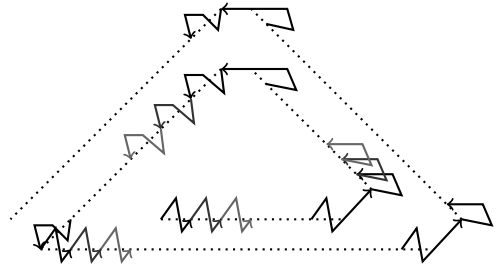

\begin{definition}[Simple trajectory] Let  $\cal{H}$ be a $\ml$.   We say that $\tau\in\Trj^s(\cal{H})$ is a simple  trajectory if  $\sigma(\traj').\sigma(\traj')\not\sqsubset\sigma(\traj)$
 for each $\traj'\in\Trj^s(\cal{H})$ such that $\sigma(\traj')\sqsubset\sigma(\tau)$ and $\traj'\neq \epsilon$.  We denote by $\STrj(\cal H)$ the set of all simple trajectries and by 
$\ssgn(\cal{H})$ the set of signatures of all simple  trajectories of $\cal{H}$.  
\end{definition}

For each multi-linear system,  the number of polyhedra in the corresponding  polyhedral partition is finite. Hence, we conclude that the number of signatures corresponding to the simple trajectories is also finite. 

\begin{lemma}\label{lemma:finite_simple_signature}  For each $\ml$  $\cal{H}$,  $\ssgn(\cal{H})$ is a finite set.  
\end{lemma}

The  notion of a cycle plays an important role in the next  section. Due to the finiteness of the  number of polyhedra in the polyhedral partition of each multi-linear system, each trajectory either reaches a region it never leaves or its subtrajectories  form cycles of boundary elements. 

\begin{definition}[Cycle]  Let  $\cal{H}$ be a $\ml$. We say that  $\traj$, a (sub)trajectory of  $\cal H$,  is a cycle if $\sign(\traj)=\sigma^{\seq}$ for $\sigma\in\ssgn(\cal H)$. We denote by $\Cycle(\cal H)$ the set of all cycles of $\cal H$.
 
\end{definition}

Multi-linear systems for the dimension two have a nice property that makes the analysis  simpler: 
Each trajectory has an ultimately periodic structure, i.e. after finite number of steps
it forms a cycle in terms of visited boundary elements. This property does not hold for higher dimensions.  Therefore, we introduce a   notion of a hypercycle. This  is a generalization of a cycle in the following sense: a hypercycle contains (several) cycles adjoined by simple trajectories. In each iteration of the hypercycle the number of passes through each  cycle may vary but the sequence of visited boundary elements is preserved. 

\begin{definition}[Hypercycle] Let ${\cal{H}}=({\cal P},f)$ be a multi-linear system.
We say that a trajectory $\traj$  is a hypecycle  if $\sigma(\tau)=(\sigma'_1.\sigma^{\seq}_1.\dots .\sigma'_m.\sigma^{\seq}_m)^{\seq}$ for $m\geq 1$, $\sigma_i\in\ssgn(\cal H)$  and $\sigma_i'\in\ssgn(\cal H)\cup\{\epsilon\}$ and at least one of the following holds.
\begin{itemize}
 \item There is $1\leq i\leq m$ such that $\sigma'_i\neq \epsilon$,
\item $m\geq 2$.
\end{itemize}

We denote by $\Hc(\cal H)$ the set of all hypercycles of $\cal H$.
\end{definition}

In fact, the notion of a hypercycle can be generalized further by considering cycles of hypercycles. But in this paper we restrict the class of systems under consideration to the systems such that each trajectory is either a cycle or a hypercycle after finite number of steps.

\section{Deciding Reachability for a Special Class of Multi-Linear Systems }\label{sec:reachability}

In this section we analyze topological properties of a subclass of multi-linear systems. This subclass is defined by a generalization of properties of $2$-dimensional \ml~systems. 

 Namely, we assume  that if there is a cyclic trajectory, then  the points of each cycle iteration intersecting the same boundary element of a polyhedron lie on a straight line, called the $\lambda$-property. A straightforward consequence of this assumption is that the distances between the corresponding boundary points of different rounds are proportional. 

If there is a cyclic trajectory, then  the points of each cycle iteration intersecting the same boundary element of a polyhedron not necessarily lie on a straight line. In general case, even the angle  between the corresponding line segments is not preserved. Nevertheless, we tend to think that for sufficiently many systems, especially for systems having some symmetry in their description,   the trajectories obey the $\lambda$-property.

\subsection{The Reachability Problem}

In the following to be able to perform exact computations, we assume that all coefficients in a system are rationals.

Since a  solution of a differential equation is unique for a given initial point in combination with rationality of coefficients, we obtain the following property. Given a multi-linear system and a rational  initial  point $x$, it is possible to compute  the point $y$ 
reachable  from $x$ after time interval $\Delta t$ exactly.  

\begin{definition}[Reachability problem]
 Given a multi-linear system ${\cal{H}}=({\cal P},f)$ and two points $x$ and $y$, the problem of point-to-point reachability  $\reach({\cal{H}},x,y)$ is stated as follows: Given two points $x,y\in {\cal R}^3$, is there a trajectory $\traj({\cal{H}},x)$ such that   $y\in\traj({\cal{H}},x)$.
\end{definition}

\subsection{Reachability for Multi-Linear $\lambda$-Systems}

Now we define formally a subclass of multi-linear systems we consider. 


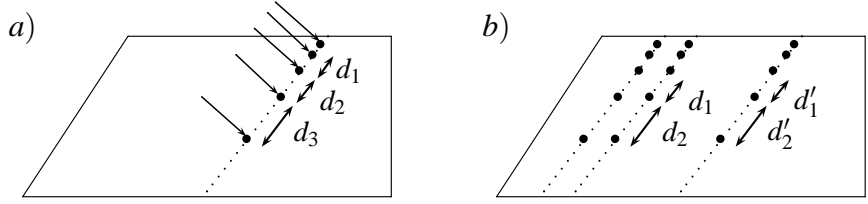
\begin{figure}
\begin{center}

 \psset{xunit=0.7,yunit=0.7}
\begin{pspicture}(0,0)(16,4.5)

\rput(0,3.6){\large $a)$}

\rput(9,3.6){\large $b)$}

\psline[linewidth=0.015,arrows=c-c](0,0.4)(2,3.45)
\psline[linewidth=0.015,arrows=c-c](7,0.4)(7,3.45)
\psline[linewidth=0.015,arrows=c-c](0,0.4)(7,0.4)
\psline[linewidth=0.015,arrows=c-c](2,3.45)(7,3.45)

\psline[linestyle=dotted,arrows=c-c](3.5,0.5)(5.8,3.45)

\psline[arrows=<->](4.54,1.35)(5.13,2.129)
\psline[arrows=<->](5.2,2.18)(5.55,2.62)
\psline[arrows=<->](5.6,2.65)(5.9,3.07)

\rput(6.2,2.8){$d_1$}
\rput(5.9,2.2){$d_2$}
\rput(5.4,1.6){$d_3$}

\psellipse*(4.25,1.5)(0.08,0.08)
\psellipse*(4.9,2.3)(0.08,0.08)
\psellipse*(5.25,2.8)(0.08,0.08)
\psellipse*(5.5,3.1)(0.08,0.08)

\psellipse*(5.65,3.3)(0.08,0.08)

\psline[linewidth=0.02,arrows=c->](3.4,2.3)(4.245,1.55)
\psline[linewidth=0.02,arrows=c->](4.05,3.1)(4.85,2.35)

\psline[linewidth=0.02,arrows=c->](4.4,3.6)(5.245,2.85)
\psline[linewidth=0.02,arrows=c->](4.65,3.9)(5.45,3.15)
\psline[linewidth=0.02,arrows=c->](4.8,4.1)(5.645,3.35)

\psline[linewidth=0.015,arrows=c-c](9,0.4)(11,3.45)
\psline[linewidth=0.015,arrows=c-c](16,0.4)(16,3.45)
\psline[linewidth=0.015,arrows=c-c](9,0.4)(16,0.4)
\psline[linewidth=0.015,arrows=c-c](11,3.45)(16,3.45)

\psline[linestyle=dotted,arrows=c-c](9.9,0.5)(12.2,3.45)
\psline[linestyle=dotted,arrows=c-c](10.5,0.5)(12.8,3.45)
\psline[linestyle=dotted,arrows=c-c](12.5,0.5)(14.8,3.45)

\psellipse*(10.65,1.5)(0.08,0.08)
\psellipse*(11.3,2.3)(0.08,0.08)
\psellipse*(11.7,2.8)(0.08,0.08)
\psellipse*(11.9,3.1)(0.08,0.08)
\psellipse*(12.05,3.3)(0.08,0.08)

\psellipse*(11.25,1.5)(0.08,0.08)
\psellipse*(11.9,2.3)(0.08,0.08)
\psellipse*(12.3,2.8)(0.08,0.08)
\psellipse*(12.5,3.1)(0.08,0.08)
\psellipse*(12.65,3.3)(0.08,0.08)



\psellipse*(13.25,1.5)(0.08,0.08)
\psellipse*(13.9,2.3)(0.08,0.08)
\psellipse*(14.25,2.8)(0.08,0.08)
\psellipse*(14.5,3.1)(0.08,0.08)

\psellipse*(14.65,3.3)(0.08,0.08)

\psline[arrows=<->](13.54,1.35)(14.13,2.129)
\psline[arrows=<->](14.2,2.18)(14.55,2.62)

\rput(14.9,2.2){$d'_1$}
\rput(14.4,1.6){$d'_2$}

\psline[arrows=<->](11.54,1.35)(12.13,2.129)
\psline[arrows=<->](12.2,2.18)(12.55,2.62)

\rput(12.9,2.2){$d_1$}
\rput(12.4,1.6){$d_2$}

\end{pspicture}\caption{\label{prop}  
 $\lambda$-cycle:  a) $d_1/d_2=d_2/d_3$, and  b) $d_1/d_2=d'_1/d'_2$}
\end{center}
\end{figure}


\begin{definition}[$\lambda$-cycle and $\lambda$-line]\label{def:lambda-property} Let ${\cal{H}}=({\cal P},f)$ be a \ml. Suppose for  $\traj\in\Trj^s$ the following holds.
\begin{itemize} 
\item $\traj=(x_{1}^j.\dots.x_s^j)_{j=1}^t$  where for   $1\leq i \leq  s$, $1\leq j\leq t$ there are
$P_{i}\in {\cal P}$ and  $\be_i\in\bd(P_i)$  such that  $x_{i}^j\in \be_i$.
\item $\overrightarrow{x_i^jx_i^{j+1}}=\lambda_i \cdot \overrightarrow{x_i^{j-1}x_i^j}$ for  $x_i^{j-1},x_i^j,x_i^{j+1}$,  $1\leq i\leq s$, $1<j<t$.

\end{itemize} 
Then we say that  $\traj$  is a $\lambda$-cycle.  We say that  a line $\lin_i$ is a $\lambda$-line of $\traj$ with respect to $\be_i$ if  $x_i^{j-1},x_i^j,x_i^{j+1}\in\lin_i$, $1\leq i\leq s$, $1<j<t$. 
\end{definition}

The notion of $\lambda$-cycle  can be extended to a hypercycle.  In the following, given two parallel lines $\lin_1$ and $\lin_2$, we denote by $\dist(\lin_1,\lin_2)$ the distance between $\lin_1$ and $\lin_2$, i.e. the length of a line segment $[x_1,x_2]$ such that $x_1\in\lin_1$, $x_2\in\lin_2$  and $[x_1,x_2]\perp\lin_1$.

\begin{figure}
\begin{center}

 \psset{xunit=0.7,yunit=0.70}
\begin{pspicture}(0,0)(36,4.5)

\psline[linewidth=0.015,arrows=c-c](7,0.4)(11,3.45)
\psline[linewidth=0.015,arrows=c-c](16,0.4)(16,3.45)
\psline[linewidth=0.015,arrows=c-c](7,0.4)(16,0.4)
\psline[linewidth=0.015,arrows=c-c](11,3.45)(16,3.45)

\rput(12.45,3.9){\small $d_1$}
\rput(13.25,3.9){\small $d_2$}
\rput(14.55,3.9){\small $d_3$}

\psline[linewidth=0.02,arrows=<->](12,3.6)(12.5,3.6)
\psline[linewidth=0.02,arrows=<->](12.6,3.6)(13.5,3.6)
\psline[linewidth=0.02,arrows=<->](13.6,3.6)(15.5,3.6)

\psline[linewidth=0.05,linestyle=dotted,arrows=c-c](9.5,0.5)(12,3.43)
\psline[linewidth=0.05,linestyle=dotted,arrows=c-c](10,0.5)(12.5,3.43)
\psline[linewidth=0.05,linestyle=dotted,arrows=c-c](11,0.5)(13.5,3.43)
\psline[linewidth=0.05,linestyle=dotted,arrows=c-c](13,0.5)(15.5,3.43)

\end{pspicture}\caption{\label{fig:lambda_hypercycle}  
 $\lambda$-hypercycle:  The ratio of the distances between $\lambda$-lines of consequitive rounds of a hypercycle is preserved: $d_1/d_2=d_2/d_3$}
\end{center}
\end{figure}
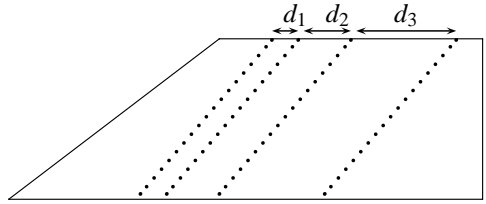

\begin{definition}[$\lambda$-hypercycle] Let ${\cal{H}}=({\cal P},f)$ be a \ml.
 Suppose a trajectory $\tau$ is a hypercycle, i.e.  $\sigma(\tau)=(\sigma'_1.\sigma^{\seq}_1.\dots .\sigma'_m.\sigma^{\seq}_m)^{\seq}$ for $\sigma_i,\sigma_i'\in\ssgn(\cal H)$. Let for $1\leq i\leq m$,   $\be\in\sigma_i$,  and $\lin_{b-2}$, $\lin_{b-1}$, $\lin_{b}$ and $\lin_{b+1}$  are $\lambda$-lines for $e$ of the corresponding consecutive rounds $b-2$, $b-1$, $b$ and $b+1$ of $\tau$.  We say that $\tau$ is a $\lambda$-hypercycle if 
\[\frac{\D_{b-1}}{\D_{b}}=\frac{\D_{b}}{\D_{b+1}},\] where $\D_a=\dist(\lin_{a-1},\lin_{a})$, $b-2<a\leq b+1$.
\end{definition}

\begin{definition}[$\lambda$-system]  Let ${\cal{H}}=({\cal P},f)$ be a \ml. We say that ${\cal{H}}$ is a $\lambda$-system if for each $\traj$ the following holds: 
1)  If $\traj$ is a cycle then $\traj$ is a $\lambda$-cycle.
2)  If $\traj$ is a hypercycle then $\traj$ is a $\lambda$-hypercycle.
\end{definition}

As the next step, we  define computable properties that would allow us to check whether a system is a $\lambda$-system.
Lemma \ref{prop:traj} defines conditions sufficient for a cycle to be a  $\lambda$-cycle:  As soon as three points of consecutive  cycle iterations  lie on a straight line, the ratio of the distances between  consecutive points of different rounds is preserved. 

\begin{lemma}\label{prop:traj} Let ${\cal{H}}=({\cal P},f)$ be a  \ml. Suppose $\sigma(\traj)=\sigma^{\seq}$ for $\traj\in\Trj^s$ and  $\sigma\in\ssgn(\cal H)$. Assume that $x_1,x_2,x_3\in\be$,  $\be\in\sigma$,  are consecutive points of intersection of $\traj$ and $\be$. If $x_1,x_2,x_3\in\lin$ for some line $\lin$ then  $\traj$ is a $\lambda$-cycle. 
\end{lemma}

As we see in Lemma \ref{lem:lambda-prop}, it is sufficient  for the $\lambda$-property to hold for two distinct trajectories going through the same  cycle of boundary elements.  Then it holds for each cycling trajectory
going through the same cycle. Note that it is sufficient to compute for each trajectory whether three points of the consecutive rounds are in one line.

\begin{lemma}\label{lem:lambda-prop} Let ${\cal{H}}=({\cal P},f)$ be a \ml. Suppose $\sigma(\traj_i)=\sigma^{\seq}$ for a trajectory $\traj_i$, $1\leq i\leq 3$.
Assume that  $\traj_1$ and $\traj_2$ are $\lambda$-cycles. Then the following holds.
\begin{itemize}
 \item $\traj_3$ is a $\lambda$-cycle. 
\item $\lin^e_1~||~\lin^e_2~||~\lin^e_3$, where $\lin^e_i$ is a  $\lambda$-line of $\traj_i$ with respect to $\be$ for each $\be\in\sigma$. 
\end{itemize} 
\end{lemma}

\begin{lemma}\label{lemma:lambda_hypercycle} Let ${\cal{H}}=({\cal P},f)$ be a \ml.  Suppose  each $\tau\in\Cycle(\cal H)$  is a $\lambda$-cycle. Then each $\tau'\in\Hc(\cal H)$  is a $\lambda$-hypercycle. 
\end{lemma}

\begin{theorem}
\label{theorem:criterion} Let ${\cal{H}}=({\cal P},f)$ be a multi-linear system. Then  it is decidable whether ${\cal{H}}$ is a $\lambda$-system.
\end{theorem}

\begin{proof}
 By Lemma  \ref{lemma:finite_simple_signature}, the set $\ssgn(\cal H)$ is finite.  By  Lemmas \ref{prop:traj}, \ref{lem:lambda-prop} and \ref{lemma:lambda_hypercycle} it is sufficient to perform the following steps:
1) For each $\sigma\in\ssgn$ to choose two distinct $\tau_1$ and $\tau_2$ such that $\sigma(\tau_1)=\sigma^{\seq}$ and  $\sigma(\tau_2)=\sigma^{\seq}$.  2) To compute  three consecutive points of intersection of $\tau_1$ and $\tau_2$ with $\be\in\sigma$. 3) To check whether these points are  in one line.    
\end{proof}

We have shown that if each  cycle is a $\lambda$-cycle then the given system is a $\lambda$-system. The algorithm to check whether a (hyper)cycle is infinite is presented in the next section and it  is an extension of the $2$-dimensional case  from \cite{MP93}.


\section{Algorithm for Point-to-Point Reachability}\label{sec:algorithm}

\begin{newalgo}{Point-to-Point Reachability}{alg:ptp_reachability}
  \INPUT points $x_0, y \in \mathcal{R}^3$, \ml~system $\mathcal{H}$, maximal simulation steps $n \in \mathcal{N}$
  \OUTPUT $\exists \traj({\cal H},x_0) = x_0,x_1, \dots, y$
  
  \medskip
  
  \STATE $y^\prime = \mathcal{H}(y)$ \COMMENT{$y \in P_i$ and $y^\prime \in \bd(P_i)$ for some partition $i$}
  \STATE $x \leftarrow x_0 \qquad k \leftarrow 0$
  
  \WHILE {$k \leq n$}
    \STATE $x \leftarrow \mathcal{H}(x) \qquad k \leftarrow k + 1$
    \IF {cycle $\cycle=(x_{i}, \dots, x_{i+s})^{\seq}$ detected}
      \IF [boundary element of $y$ is in cycle] {$\bd(y^\prime) = \bd(x_b) \in \cycle$}
	\IF  {$y^\prime = x_b^1 + t \cdot \overrightarrow{x_b^1x_b^2}$ for $t = \frac{1-\lambda_b^k}{1-\lambda_b} | k \in \mathcal{N}^+$}
	  \RETURN \TRUE
	\ENDIF 
      \ENDIF \COMMENT{$y^\prime$ is not reached by cycle $\cycle$}
      
      \IF{$\operatorname{isinfinite}(\cycle)$}
	\RETURN \FALSE
      \ELSE
	\STATE $x \leftarrow \operatorname{exitPoint}(\cycle)$
      \ENDIF
    \ENDIF \COMMENT{cycle detected}
  \ENDWHILE
  
  \RETURN $y^\prime \in \traj({\cal H},x_0)$
\end{newalgo}

\begin{newalgo}{Infinity test}{alg:infinity}
  \INPUT cycle $\cycle=(x_{i}^j, \dots, x_{i+s}^j)_{j=1}^3$
  \OUTPUT $\cycle$ is infinite cycle
  
  \medskip
  
  \FORALL{$x_b^j \in \cycle$}
    \FORALL{$e \in \bd(P_b)$}
	\IF [trace intersects \ensuremath{e}] {$e \cap \{x_b^1 + t \cdot \overrightarrow{x_b^1x_b^2 | t \in \mathcal{R}^+}\} \neq \emptyset$}
	  \IF{$\lambda_b \geq 1$}
	    \RETURN \FALSE
	  \ELSE
	    \IF{$x_b^\infty = x_b^1 + \frac{1}{1-\lambda_b} \cdot \overrightarrow{x_b^1x_b^2} \notin P_b$}
	      \RETURN \FALSE
	    \ENDIF
	  \ENDIF
	\ENDIF
    \ENDFOR
  \ENDFOR
  
  \RETURN \TRUE
\end{newalgo}
\begin{newalgo}{Exit point}{alg:exit_point}
  \INPUT cycle $\cycle=(x_{i}^j, \dots, x_{i+s}^j)_{j=1}^3$
  \OUTPUT point $x_e \in \mathcal{R}^3$ where cycle $\cycle$ is abandoned
  
 \medskip
  
  \STATE $PQ \leftarrow \mathrm{PriorityQueue}: \mathcal{N} \times \mathcal{R}^3$
  \FORALL{$x_b^j \in \cycle$}
    \IF{$\lambda_b < 1 \vee x_b^\infty \in P_b$}
      \STATE \textbf{skip}
    \ENDIF
    
    \FORALL{$e \in \bd^e(P_b)$}
	\STATE $t_e \leftarrow \frac{(\overrightarrow{x_b^1 v} \times \vec{u}) \cdot (\overrightarrow{x_b^1x_b^2} \times \vec{u})}{||\overrightarrow{x_b^1x_b^2} \times \vec{u}||^2}$
    \ENDFOR
    
    \STATE $t \leftarrow \min_{e \in \bd^e(P_b)} t_e$
    \IF{$\lambda_b \neq 1$}
      \STATE $n \leftarrow \lfloor \frac{\log(1 - t(1 - \lambda_b))}{\log(\lambda_b)} \rfloor$ 
     \STATE $t \leftarrow \frac{1-\lambda_b^n}{1-\lambda_b}$
    \ELSE
      \STATE $t \leftarrow n \leftarrow \lfloor t \rfloor$
    \ENDIF
    
    \STATE $x_e \leftarrow x_b^1 + t \cdot \overrightarrow{x_b^1x_b^2}$ 
    \STATE $PQ.\operatorname{put}(n, x_e)$
  \ENDFOR
  
  \RETURN $PQ.\operatorname{pop}()$
\end{newalgo}

First  we need to introduce some further notations. Let $e = \bd(x) | x \in \mathcal{R}^3$ denote the border element $e$ such that $x \in P_i \cap e$ for some partition $i$ and $e \in \bd(P_i)$. An edge $e$ is given in the form of \[e = \{v + \kappa \cdot \vec{u}| v, \vec{u} \in \mathcal{R}^3; l \leq \kappa \leq h :  l,h \in \mathcal{R}\}.\] Furthermore, given a cycle $\cycle=(x_{i}^j, \dots, x_{i+s}^j)_{j=1}^k$, let $x_{b}^{a}$ denote the point reached by the cycle in the $a$th iteration on border element $\bd(x_b)$. Then $\{ x_b^1 + t \cdot \overrightarrow{x_b^1x_b^2}\}$ for $t \in \mathcal{R}$ is the line through the trajectory points on $\bd(x_b)$ ($\lambda$-line), called a trace in the following.

Algorithm \ref{alg:ptp_reachability} decides for a \ml~system $\mathcal{H}$ and a starting point $x_0 \in \mathcal{R}^3$ whether a point $y \in \mathcal{R}^3$ can be reached by a trajectory $\traj({\cal H},x_0) = x_0,x_1, \dots, y$ through $\mathcal{H}$. The algorithm is allowed to perform $n \in \mathcal{N}$ simulations of $\mathcal{H}$, note that $k$ can be much larger than $n$ as our experiments will show.

While the maximum number of evaluations is not reached, $\mathcal{H}$ is simulated stepwise until either $y$ is reached or a cycle is detected. In our implementation we use the cycle detection algorithm due to Brent \cite{B80} which requires $O(\mu + \lambda)$ system evaluations.\footnote{$\mu$ denoting the first occurrence of the cycle, $\lambda$ indicating the cycle length.} 

If a cycle $\cycle=(x_{i}^j, \dots, x_{i+s}^j)_{j=1}^k$ is detected, several cases have to be distinguished: 

\begin{enumerate}

\item[a)] $\bd(y^\prime)$ is \emph{not} part of the cycle.\footnote{$y^\prime$ is the border element reached from $y$ by simulating $\mathcal{H}$, see algorithm \ref{alg:ptp_reachability} step 1.} If $\cycle$ is infinite according to Algorithm \ref{alg:infinity}, $y$ will never be reached. 

\item[b)]  $\bd(y^\prime) = \bd(x_b)$ for some $b \in [i, i+s]$ and hence is element of the cycle. In this case it needs to be checked whether $\exists a \in \mathcal{N}$ such that $y^\prime = x_b^a$. If so then $y^\prime$ should be in \[\{x_b^1 + t \cdot \overrightarrow{x_b^1x_b^2}\}\] for some $t_y \in \mathcal{R}^+$. For all $x_b^a$, \[t_a = \sum_{i = 1}^a \lambda_b^i = \frac{1-\lambda_b^i}{1-\lambda_b},\] therefore \[a_y = \frac{\log(1 - t_y(1 - \lambda_b))}{\log(\lambda_b)}\] is in $\mathcal{N}^+$ iff $y^\prime$ is reached by a cycle iteration. If $y^\prime$ is not reached by $\cycle$ and the cycle is infinite, $y$ is never reached by $\mathcal{H}$ from $x_0$.

\item[c)]  $y^\prime$ is not reached by $\cycle$ and the cycle is \emph{finite}. We calculate the point $x_e$ where $\cycle$ is abandoned according to algorithm \ref{alg:exit_point} and continue simulation and cycle detection there.

\end{enumerate}

The infinity test (Algorithm \ref{alg:infinity}) checks for every partition $P_b$ in cycle $\cycle$ whether the trace line $\{x_b^1 + t \cdot \overrightarrow{x_b^1x_b^2}\}$ intersects with some $e \in \bd(P_b)$. If a intersection point $x_{IS_{b,e}}$ exists and $\lambda_b \geq 1$ then $\cycle$ must abandon $P_b$ after some number of iterations and therefore can not be infinite. If no intersection occurs or $\lambda_b < 1$ and the convergence point \[x_b^\infty = x_b^1 + \sum_{i = 1}^\infty \lambda_b^i \cdot \overrightarrow{x_b^1x_b^2} = x_b^1 + \frac{1}{1-\lambda_b} \cdot \overrightarrow{x_b^1x_b^2}\] lies before $x_{IS_{b,e}}$ on the trace line, $P_b$ is never abandoned. If no partition within the cycle is ever abandoned then $\cycle$ is infinite.

For a finite cycle $\cycle$, Algorithm \ref{alg:exit_point} determines the exit point $x_e$ from $\cycle$. Recall that an edge can be represented as $e = \{v + \kappa \cdot \vec{u} | l \leq \kappa \leq h\}$. For every $e \in \bd(P_b)$ the intersection point $x_{IS_{b,e}}$ with the trace line is determined. 

The intersection point with the smallest distance $t$ to $x_b^1$ is the exit point to $P_b$. The number of cycle iterations fully contained in $P_b$ is given by \[n_b = \Big\lfloor \frac{\log(1 - t(1 - \lambda_b))}{\log(\lambda_b)} \Big\rfloor\] with point \[x_{e,b} = x_b^1 + \frac{1-\lambda_b^n}{1-\lambda_b} \cdot \overrightarrow{x_b^1x_b^2}.\] Special consideration is given to $\lambda_b = 1$, refer to Algorithm \ref{alg:exit_point}. The overall exit point to cycle $\cycle$ can therefore be determined by $x_e = x_{e,k}$ with $k = \operatorname{argmin}_{b \in [1,t]} n_b$.

\section{Experiments}\label{sec:experiments}

\begin{figure}[tb]
  \centering
  \subfloat[$\mathcal{H}_1$]{\includegraphics[width=.38\textwidth]{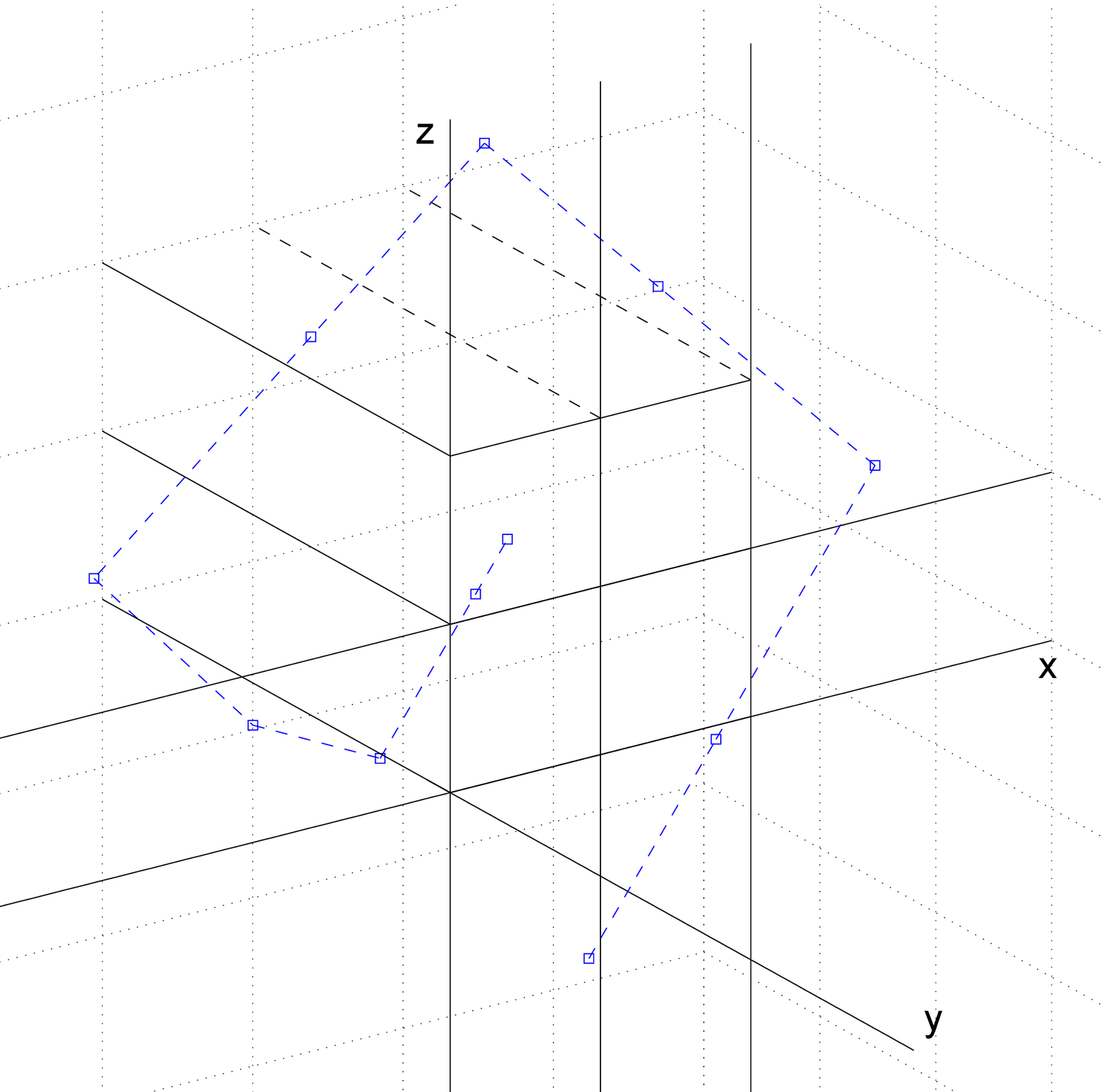}}
  \subfloat[$\mathcal{H}_2$]{\includegraphics[width=.38\textwidth]{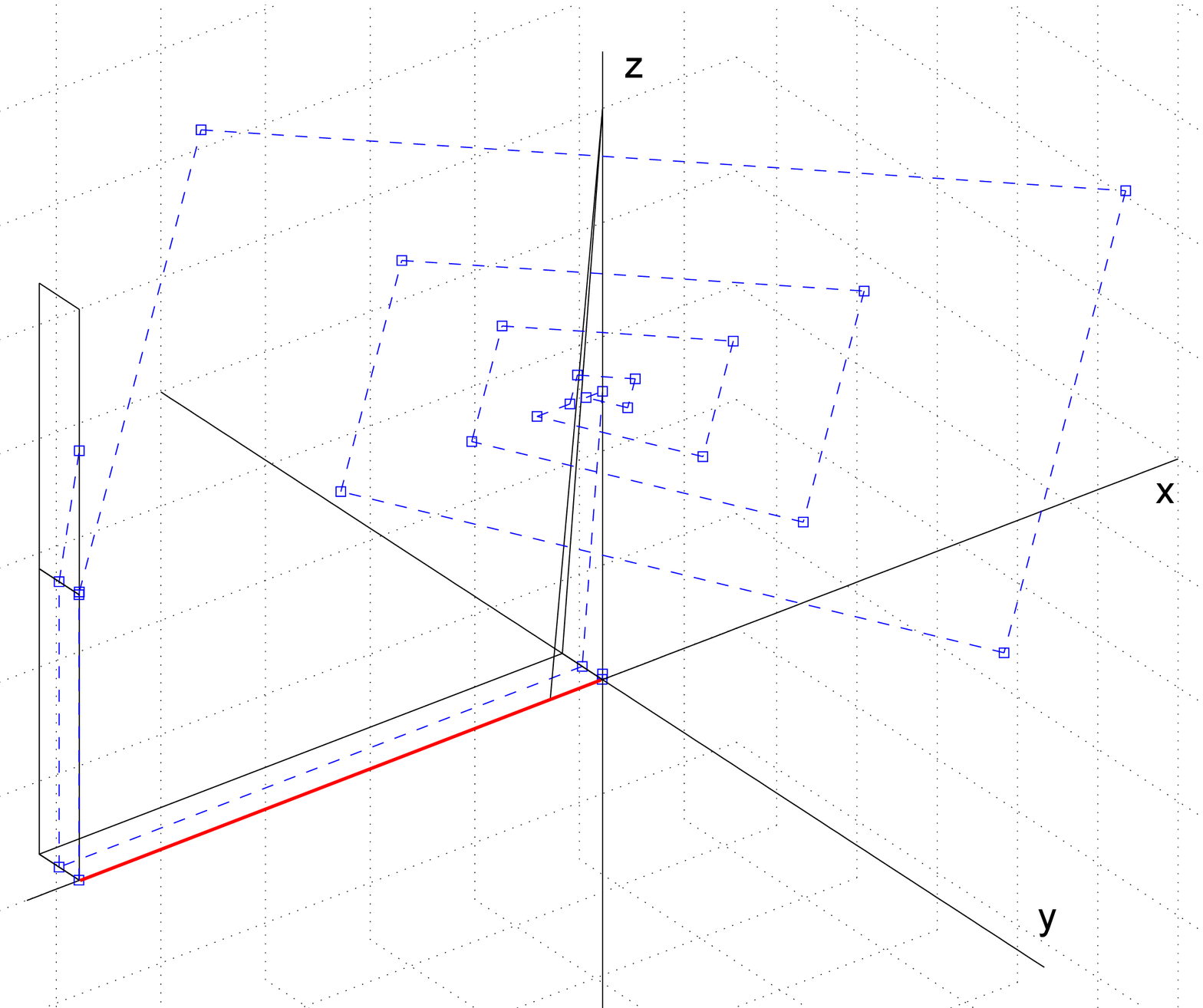}}
  \caption{Example \ml~systems with trajectories. The convergence line of $\mathcal{H}_2$ is highlighted in red.}
  \label{fig:example_systems}
\end{figure}

\begin{figure}[tb]
  \centering
  \subfloat[$xOy$ plane]{\includegraphics[width=.3\textwidth]{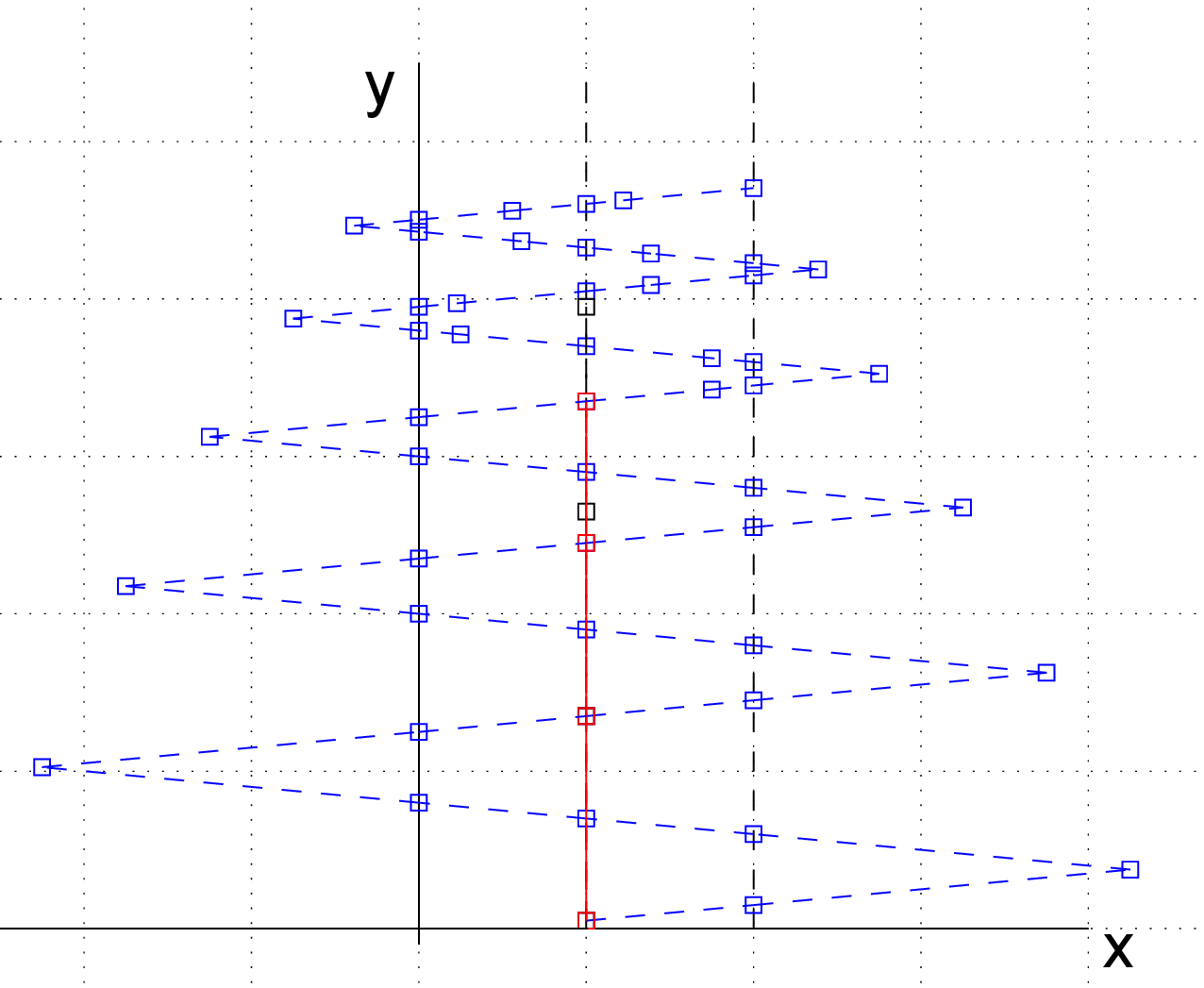}}
  \subfloat[$xOz$ plane]{\includegraphics[width=.3\textwidth]{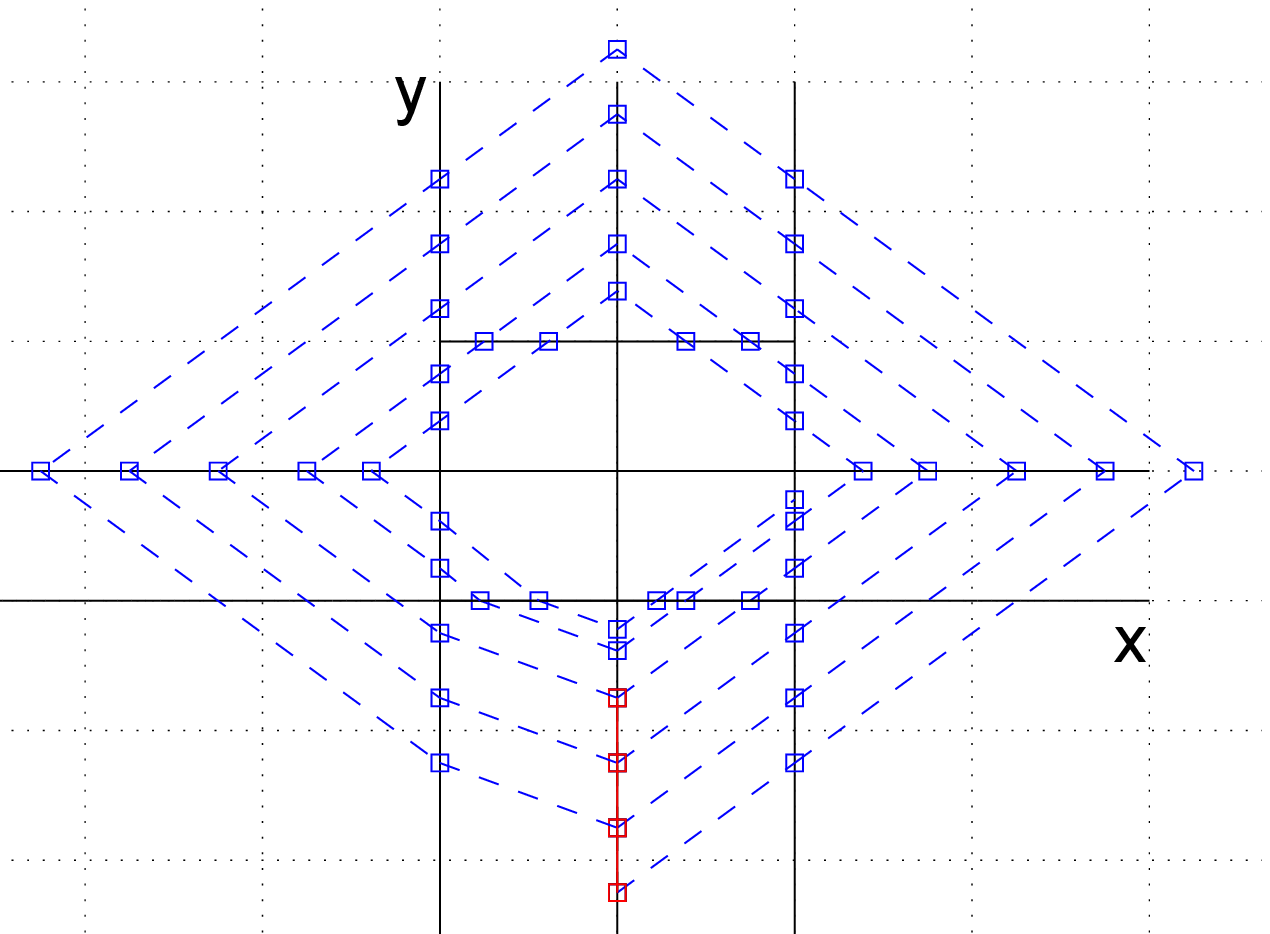}}
  \subfloat[$yOz$ plane]{\includegraphics[width=.3\textwidth]{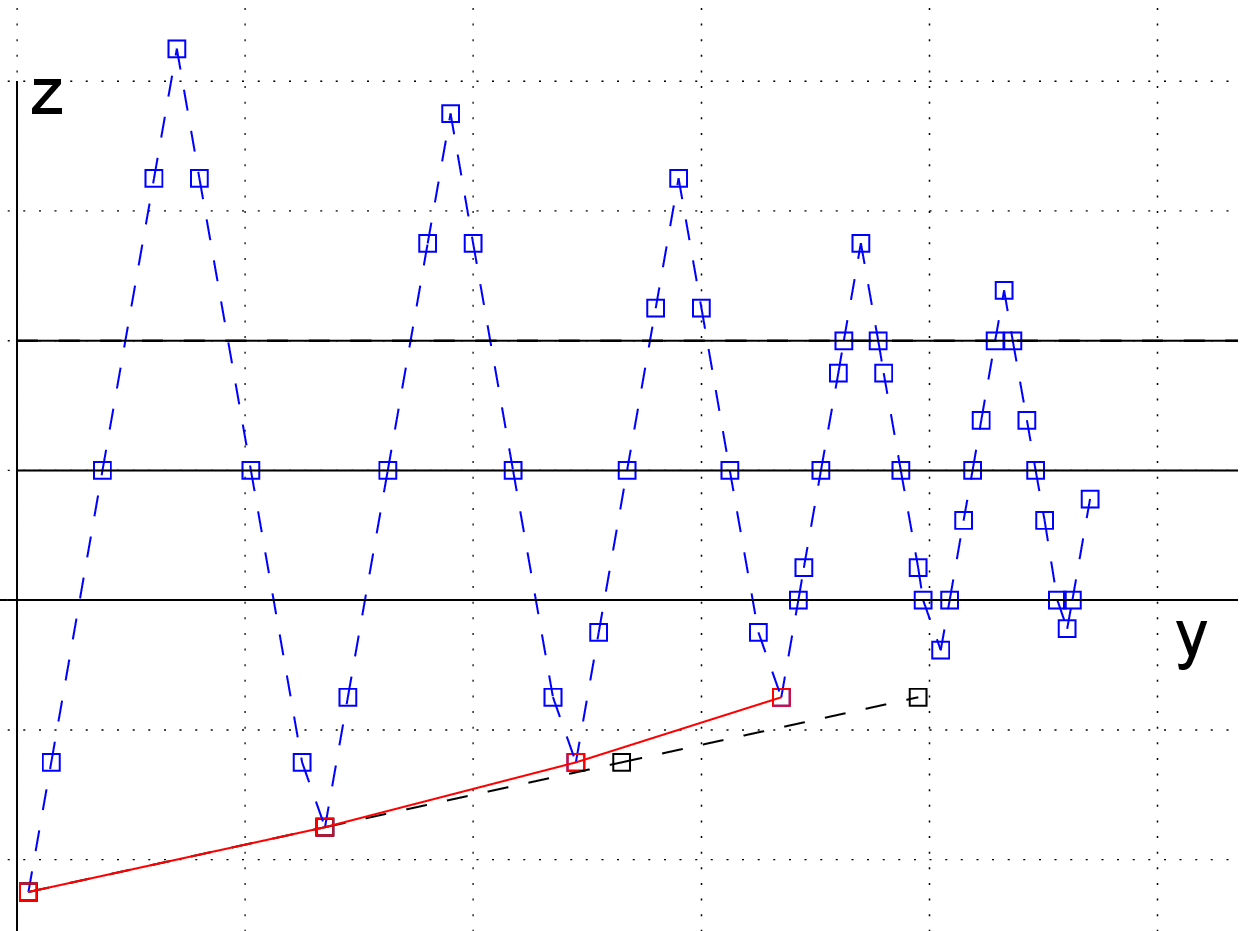}}
  \caption{Projections of trajectory in $\mathcal{H}_1$. In c), the dashed black line indicates the trajectory if the points were to be on one line.}
  \label{fig:projections1}
\end{figure}

\begin{figure}[tb]
  \centering
  \includegraphics[width=.28\textwidth]{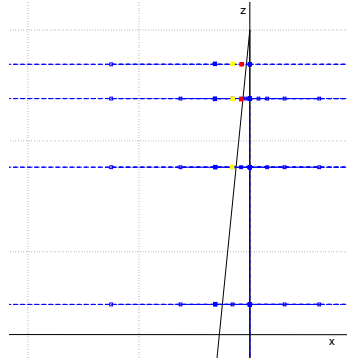}
  \caption{Projection on $xOz$ plane of trajectory in $\mathcal{H}_2$. With each iteration of the hyper cycle, more iterations of the simple cycle through the four 3 dimensional polyhedra are required to reach $P_1$. Therefore no infinite simple cycle exits in $\mathcal{H}_2$.}
  \label{fig:projections2}
\end{figure}

The  experiments were performed on a server with two dual-core 2.8 GHz CPUs and 3 GB main memory under RedHat Linux. We implemented the algorithms in Java using the JAMA library for linear algebra operations \cite{HetAl00}. 
The source code of our implementation is available at \cite{TF11}.

Two sample \ml~systems were used in our experiments and are depicted in Figure \ref{fig:example_systems} alongside a sample trajectory for each.  The  details of the examples can be found in \cite{TF11-2}.

All partitions in $\mathcal{H}_1$ are unbounded in the $Oy$ dimension. The $xOz$ plane is divided into four inner partitions and 8 outer ones. In the inner partitions the trajectory \enquote{rotates} around the center with increasing radius, whereas the radius decreases in the outer partitions. 
Not all choices of $\mathbf{y}$ result in a $\lambda$-system. Note that in figure \ref{fig:projections1} the points of a trace through one border element (depicted in red) lie on one line for 2 projections ($xOy$ and $xOz$), but violate the line criterion for the $yOz$ projection.

System $\mathcal{H}_2$ consists of 5 two dimensional polygons and 4 unbounded three dimensional polyhedra. In the unbounded regions the system \enquote{rotates} around the $Oz$ axis in ever shrinking circles until it reaches $P_1$. $\mathcal{H}_2$ then traverses all two dimensional polygons until it reaches $P_5$ where it is ejected into to the unbounded space again. With each iteration more and more rotations are required in the unbounded space to reach $P_1$ and the trajectory through $P_3$ converges towards the $Ox$ axis. Therefore the system never reaches an infinite simple cycle as illustrated in figure \ref{fig:projections2}.

In $\mathcal{H}_1$ placing the initial point $x_0$ at any distance $n$ from the inner four partitions results in $\Theta(n)$ simulation steps until the inner portion is reached. Our algorithm reduces the complexity to $O(\mu + \lambda)$ in general, considering $\mathcal{H}_1$ even to $O(1)$. Only three cycle iterations are required to calculate $\lambda_b$ for each $x_b \in \cycle=(x_{i}, \dots, x_{i+s})^{\seq}$ and determine the exit point of the cycle. Experimental data is shown in table \ref{tab:experiment} and figure \ref{fig:experiment:a}. We attribute the decrease in running time of our algorithm in the first two iterations to the Java just in time compiler, optimizing code dynamically as it is executed \cite{CetAl97}. Thereafter the algorithm exhibits constant execution time as anticipated.

Modifying $\mathcal{H}_1$ to $\mathcal{H}_1^\prime$ so that the rotation radius decreases in the outer partitions as well as in the inner partitions,\footnote{Specifically setting $c_1$ to $(\frac{1}{2},y_1,-1)$.} produces a convergence line for all partitions at $Cl = \{(C,\cdot, C)\}$. Simulation alone may \emph{never} determine whether $y \in Cl$ is reached, whereas our algorithm requires again three cycle iterations of length at most 12 to determine the reachability of $y$.

\ml~system $\mathcal{H}_2$ exhibits a similar behavior. With each pass through the two dimensional partitions the number of required rotations to reach the $Oz$ axis again increases to infinity. Therefore the reachability of $y$ on or close to the convergence line of the system is not feasibly determined by simulation alone. If $y$ is reached after $n$ hyper cycle iterations, at least $\Theta(n^2)$ simulation steps were required. Our algorithm reduces the complexity to $O(n(\mu + \lambda))$. With improved hyper cycle handling the complexity ought to be further reduced to $O(c \cdot (\mu + \lambda))$, since three passes of the simple cycle to reach the $Oz$ axis suffice to determine the convergence line of the hyper cycle. Again, experimental data is shown in table \ref{tab:experiment} and figure \ref{fig:experiment:b}. The data exhibits the same behavior as for $\mathcal{H}_1$ regarding the Java just in time compilation. Due to its position $x_0 = (10^8, 10^8, C - C/10^8)$ reaches $P_2$ before $P_1$ and therefore requires fewer simulation steps than $x_0 = (10^7, 10^7, C - C/10^7)$.

\begin{table}[tb]
  \centering
  \begin{footnotesize}
  \subfloat[$\mathcal{H}_1$: $y$ is first point to be reached in the inner four partitions.]{
  \begin{tabular}{|l|r|r|r|}
  \hline
  $x_0 = (5, 0, -x)$ & \multirow{2}{*}{\textbf{Steps}} & \multirow{2}{*}{\textbf{Simulation}} & \multirow{2}{*}{\textbf{PTPR}} \\
  \cline{1-1}
  $x$ & & & \\
  \hline  
  $10^1$ & 10 & 10 & 24 \\ 
  $10^2$ & 154 & 98 & 18 \\ 
  $10^3$ & 1594 & 1421 & 6 \\ 
  $10^4$ & 15994 & 2082 & 6 \\ 
  $10^5$ & 159994 & 22609 & 5 \\ 
  $10^6$ & 1599994 & 188276 & 5 \\ 
  $10^7$ & 15999994 & 1952252 & 5 \\ 
  $10^8$ & 159999994 & 19700805 & 7 \\
  \hline
  \end{tabular}}
  \subfloat[$\mathcal{H}_2$: $y = (0, 0, C - C/x)$.]{
  \begin{tabular}{|l|r|r|r|}
  \hline
  $x_0 = (x, x, C - C/x)$ & \multirow{2}{*}{\textbf{Steps}} & \multirow{2}{*}{\textbf{Simulation}} & \multirow{2}{*}{\textbf{PTPR}} \\
  \cline{1-1}
  $x$ & & & \\
  \hline  
  $10^1$ & 10 & 10 & 12 \\
  $10^2$ & 59 & 20 & 13 \\
  $10^3$ & 83 & 17 & 3 \\
  $10^4$ & 111 & 23 & 4 \\
  $10^5$ & 135 & 26 & 4 \\
  $10^6$ & 149 & 30 & 3 \\
  $10^7$ & 161 & 33 & 3 \\
  $10^8$ & 143 & 29 & 3 \\
  \hline
  \end{tabular}}
  \end{footnotesize}
\caption{Comparison of our algorithm to pure simulation to decide reachability of $y$ given $x_0$. Column \emph{steps} lists the number of simulation steps required to reach $y$. Columns \emph{simulation} and \emph{PTPR} give the time in [ms] required to decide reachability by simulation and our algorithm, respectively.}
\label{tab:experiment}
\end{table}

\begin{figure}[tb]
  \centering
  \subfloat[$\mathcal{H}_1$]{\includegraphics[width=.42\textwidth]{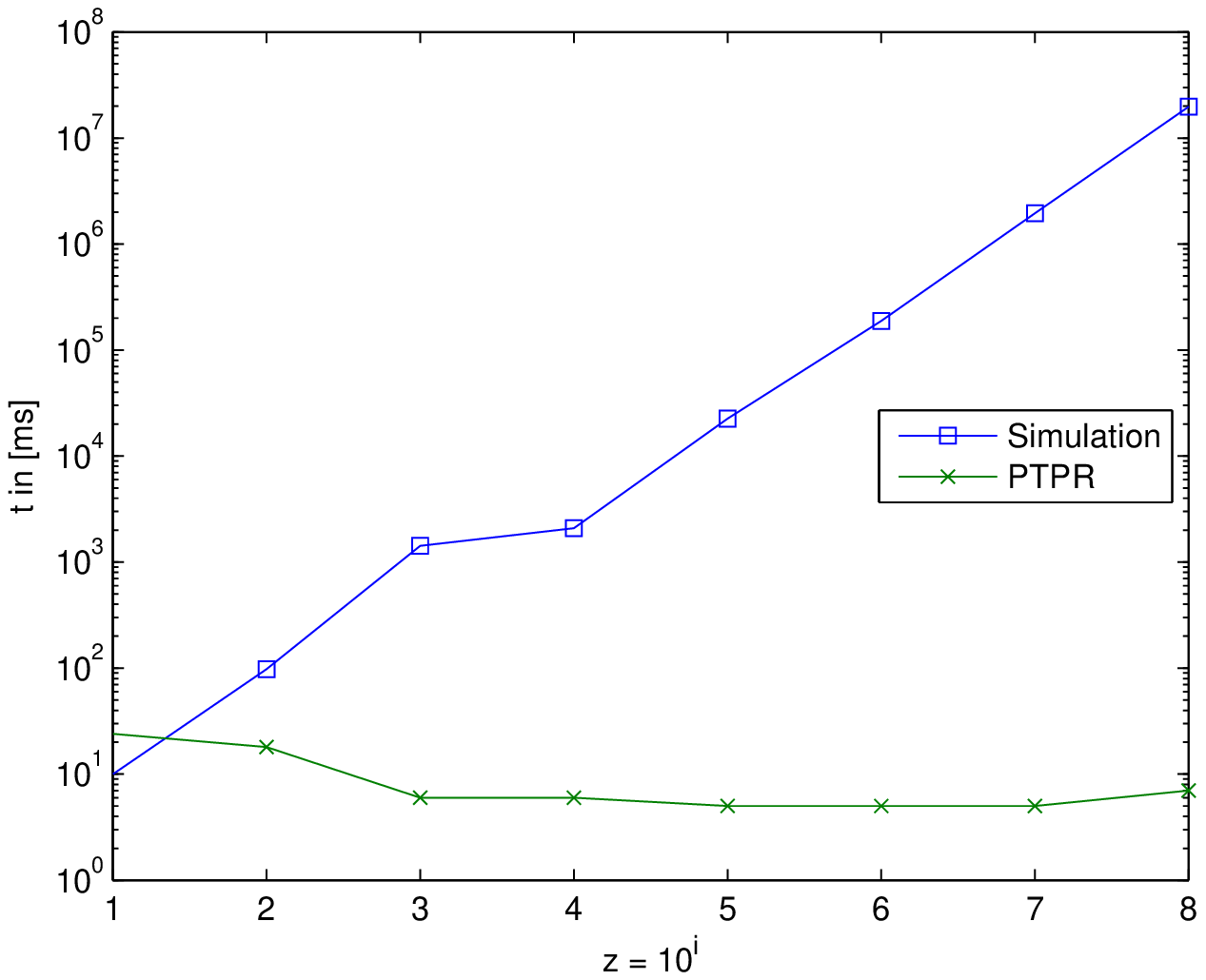}\label{fig:experiment:a}}
  \subfloat[$\mathcal{H}_2$]{\includegraphics[width=.42\textwidth]{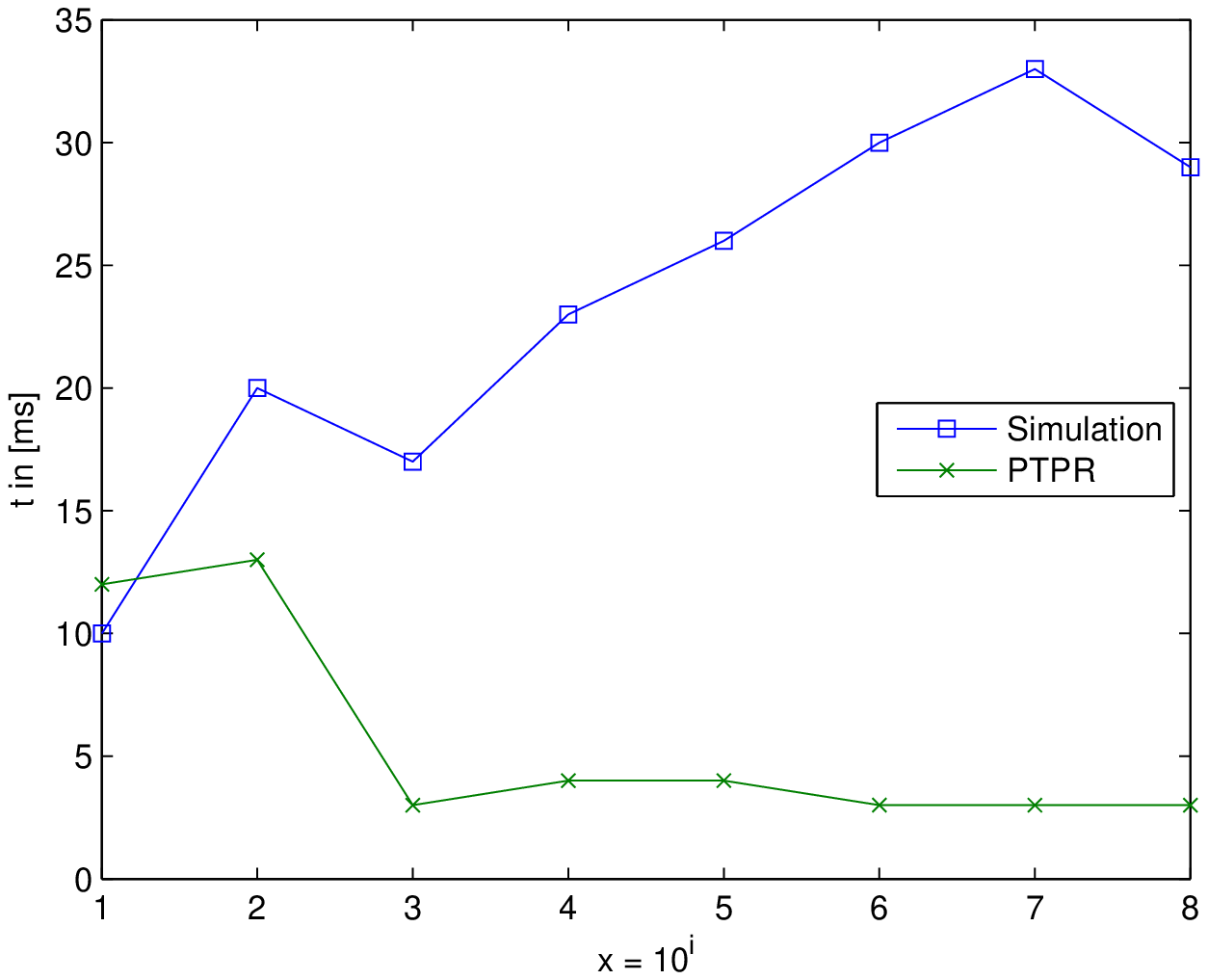}\label{fig:experiment:b}}
  \caption{Comparison of our algorithm (\emph{PTPR}) to pure simulation to decide reachability of $y$ given $x_0$.}
  \label{fig:experiment}
\end{figure}

\section{Conclusions}\label{sec:conclusions}

The complexity of safety critical systems has increased dramatically over last decades. 
The safety properties of such systems  can often not be checked exactly  either due to theoretical boundaries or due to too large
computational efforts required. One of the drawbacks of recent techniques is that too little attention is paid to the geometric properties of the systems under analysis.

A hybrid system (a hybrid automaton) is a formalism that can be used  for modeling safety critical systems. \ml~systems constitute a rather simple class of hybrid systems but yet they are on the boundary of decidable and undecidable systems.  

\ml~systems have certain properties that make them very suitable for a topological analysis. 
 We have shown that on the one hand there are systems with acyclic behavior, and on the other hand if
some properties of $2$-dimensional systems hold in three dimensions then it is possible to answer  exactly  whether a point $y$ is reachable from a point $x$.   
We have presented a prototype implementation of our approach for solving the reachability problem for a subclass of multi-linear systems which we have called $\lambda$-systems.
We compared our approach  with simulation. The results  suggest, that using geometrical properties of the systems can lead  to  orders of magnitude more efficient techniques than simulation.
As soon as our algorithm detects a cycle (or a hypercycle) for which the infinity criterion holds, the algorithm requires constant number of steps. While the number steps for simulation grows exponentially with the distance between points. 
Also  our algorithm can lead to more exact computations because of  less rounding errors during the computation. 


\bibliographystyle{packages/eptcs}

\bibliography{HybridSystems}

\end{document}